\documentclass[twocolumn,showpacs,superscriptaddress,nofootinbib]{revtex4-2}
\usepackage[latin9]{inputenc}
\usepackage{amsmath}
\usepackage{amssymb}
\usepackage{amsthm}
\usepackage{graphicx}
\usepackage{color}
\usepackage{enumitem}  

\theoremstyle{plain}

\newtheorem{theorem}{Theorem}

\makeatletter

\@ifundefined{textcolor}{}
{
 \definecolor{BLACK}{gray}{0}
 \definecolor{WHITE}{gray}{1}
 \definecolor{RED}{rgb}{1,0,0}
 \definecolor{GREEN}{rgb}{0,1,0}
 \definecolor{BLUE}{rgb}{0,0,1}
 \definecolor{CYAN}{cmyk}{1,0,0,0}
 \definecolor{MAGENTA}{cmyk}{0,1,0,0}
 \definecolor{YELLOW}{cmyk}{0,0,1,0}
}

\usepackage{color}
\usepackage{amsfonts}
\usepackage{graphics}
\usepackage{bm} 
\usepackage{epsfig}\usepackage{amsthm}

\def\identity{\leavevmode\hbox{\small1\kern-3.8pt\normalsize1}}

\newcommand{\ket}[1]{\left | #1 \right\rangle}
\newcommand{\bra}[1]{\left \langle #1 \right |}

\renewcommand{\epsilon}{\varepsilon}

\makeatother

\begin{document}

\title{Correlations and energy in mediated dynamics}

\author{Tanjung Krisnanda}
\affiliation{School of Physical and Mathematical Sciences, Nanyang Technological University, 637371 Singapore, Singapore}
%\blfootnote{Correspondence: }

\author{Su-Yong Lee}
\email{Current address: Agency for Defense Development, Daejeon 34186, Korea}
\affiliation{School of Computational Sciences, Korea Institute for Advanced Study, Hoegi-ro 85, Dongdaemun-gu, Seoul 02455, Korea}

\author{Changsuk Noh}
\affiliation{Kyungpook National University, Daegu 41566, Korea}

\author{Jaewan Kim}
\affiliation{School of Computational Sciences, Korea Institute for Advanced Study, Hoegi-ro 85, Dongdaemun-gu, Seoul 02455, Korea}

\author{Alexander Streltsov}
\affiliation{Centre for Quantum Optical Technologies, Centre of New Technologies, University of Warsaw, 02-097 Warsaw, Poland}

\author{Timothy C. H. Liew}
\affiliation{School of Physical and Mathematical Sciences, Nanyang Technological University, 637371 Singapore, Singapore}
\affiliation{MajuLab, International Joint Research Unit UMI 3654, CNRS, Universit\'{e} C\^{o}te d'Azur, Sorbonne Universit\'{e}, National University of Singapore, Nanyang Technological University, Singapore}

\author{Tomasz Paterek}
\affiliation{Institute of Theoretical Physics and Astrophysics, Faculty of Mathematics, Physics and Informatics, University of Gda\'{n}sk, 80-308 Gda\'{n}sk, Poland}

\begin{abstract}
The minimum time required for a quantum system to evolve to a distinguishable state is set by the quantum speed limit, and consequently influences the change of quantum correlations and other physical properties.
Here we study the time required to maximally entangle two principal systems interacting either directly or via a mediating ancillary system, under the same energy constraints.
The direct interactions are proved to provide the fastest way to entangle the principal systems, but it turns out that there exist mediated dynamics that are just as fast.
We show that this can only happen if the mediator is initially correlated with the principal systems.
These correlations can be fully classical and can remain classical during the entangling process.
The final message is that correlations save energy: 
one has to supply extra energy if maximal entanglement across the principal systems is to be obtained as fast as with an initially correlated mediator.
\end{abstract}

\maketitle

%\section{Introduction}
An evolution of a quantum state into a distinguishable one requires finite time.
The shortest time to achieve this task is governed by the quantum speed limit (QSL).
The first lower bound on the shortest time was derived in a pioneering work by Mandelstam and Tamm~\cite{mandelstam}.
Thereafter, important advancements and extensions of the QSL were reported, for example, for pure states~\cite{pure1,pure2,margolus} as well as mixed states~\cite{mix1,mix2,mix3,mix4}.
The applications of these fundamental findings have been valuable in many areas, e.g., in the analysis for the rate of change of entropy~\cite{deffner2010generalized}, 
the limitations in quantum metrology~\cite{giovannetti2011advances} and quantum computation~\cite{lloyd2000ultimate,lloyd2002computational}, 
and the limit on charging capability of quantum batteries \cite{binder2015quantacell,binderphdthesis,campaioli2017enhancing}.
See also Refs.~\cite{shanahan2018quantum,okuyama2018quantum} for studies showing the application of QSL in the classical regime.

The widely accepted time bound for an evolution of a quantum state $\rho$ (in general, mixed) to another state $\sigma$ is known as the \emph{unified} QSL~\cite{unifiedbound,revbound}, which reads
\begin{equation}\label{EQ_bound}
\tau(\rho,\sigma)\ge \hbar \frac{\Theta(\rho,\sigma)}{\min \{ \langle H\rangle,\Delta H\}},
\end{equation}
where $\Theta(\rho,\sigma)=\arccos( \mathcal{F}(\rho,\sigma))$ denotes a distance measure known as the Bures angle, 
$\mathcal{F}(\rho,\sigma)=\mbox{tr}\left( \sqrt{\sqrt{\rho}\sigma\sqrt{\rho}} \right)$ the Uhlmann root fidelity \cite{fidelity1,fidelity2}, 
$\langle H\rangle=\mbox{tr}(H\rho)-E_g$ the mean energy taken relative to the ground level of the Hamiltonian, $E_g$, 
and $\Delta H=\sqrt{\mbox{tr}[H^2\rho]-\mbox{tr}[H\rho]^2}$ the standard deviation of energy (SDE).
%For the case of pure states, say, $|\psi \rangle$ and $|\phi \rangle$, the Bures angle reduces to the Fubini-Study distance, which reads $\Theta(\ket{\psi}\bra{\psi},\ket{\phi}\bra{\phi})=\arccos|\langle \psi |\phi \rangle|$ \cite{FS1,FS2,FS3}.
Note also that other distances have been employed~\cite{revbound}.
In essence, Eq.~(\ref{EQ_bound}) is often described as a version of time-energy uncertainty relation as the evolution time is lower bounded by the amount of energy (mean or variance) initially accessible to the system.

Here we investigate the evolution speed of two principal objects $A$ and $B$, which interact either directly or via an ancillary system $C$.
While direct interactions place no restrictions on the joint Hamiltonian $H_{AB}$, the mediated dynamics is mathematically encoded in the assumption that the tripartite Hamiltonian 
is a sum $H_{AC} + H_{BC}$ that excludes the terms coupling $A$ and $B$ directly.
Note that local Hamiltonians, i.e., $H_A$, $H_B$, and $H_C$, are already included in these general forms.
These scenarios are quite generic and applicable to ample situations.
We are interested in contrasting them and in identifying resources different than energy that play a role in speeding up the evolution.
We therefore impose the same energy constraint (the denominator in Eq.~(\ref{EQ_bound})) in both bipartite and tripartite settings.
Under this condition we show achievable minimal time required to maximally entangle principal systems starting from disentangled states.
It turns out that the mediated dynamics cannot be faster than the direct dynamics, but it can be just as fast provided that the mediator is initially correlated with the principal systems.
We show additionally, with an explicit example, that although entanglement gain between $A$ and $B$ is the desired quantity, the correlations to the mediator can remain classical at all times, see also Refs.~\cite{krisnanda2017,pal2021experimental}.
These results can be interpreted in terms of trading correlations for energy.
If one starts with an uncorrelated mediator and aims at entangling the principal systems as fast as with a correlated mediator, additional energy has to be supplied initially.
On the other hand, due to energy conservation, the same energy must be invested in order to prepare the correlated mediator, see Refs.~\cite{Huber2015,Bruschi2015,Bakh2019} for a discussion from a thermodynamic perspective.

%%%%%%%%%%%%%%%%%%%%%%%%%%%%%%%%%%%%%%%%%%%%%%%%%%%%%%%%%%%%%%%%%%%%%%%%%%%%%%%%%%%%%%%%%%%%%%%%%%%%%%%%%%%%%%%%%%

\section{Preliminaries}
Figure~\ref{FIG_setupspsp} summarises different considered generic scenarios.
We shall refer to the case of direct interactions as $\mathcal{DI}$ and split the mediated interactions into two cases
where mediator $C$ either interacts with the principal systems at all times ($\mathcal{CMI}$ for continuously mediated interactions)
or where it first interacts with $A$ and then with $B$ ($\mathcal{SMI}$ for sequentially mediated interactions).
Note that $\mathcal{SMI}$ in particular covers the case of commuting Hamiltonians $H_{AC}$ and $H_{BC}$.
We begin by explaining the energy constraints imposed on these scenarios.

\begin{figure}[h]
\includegraphics[width=0.4\textwidth]{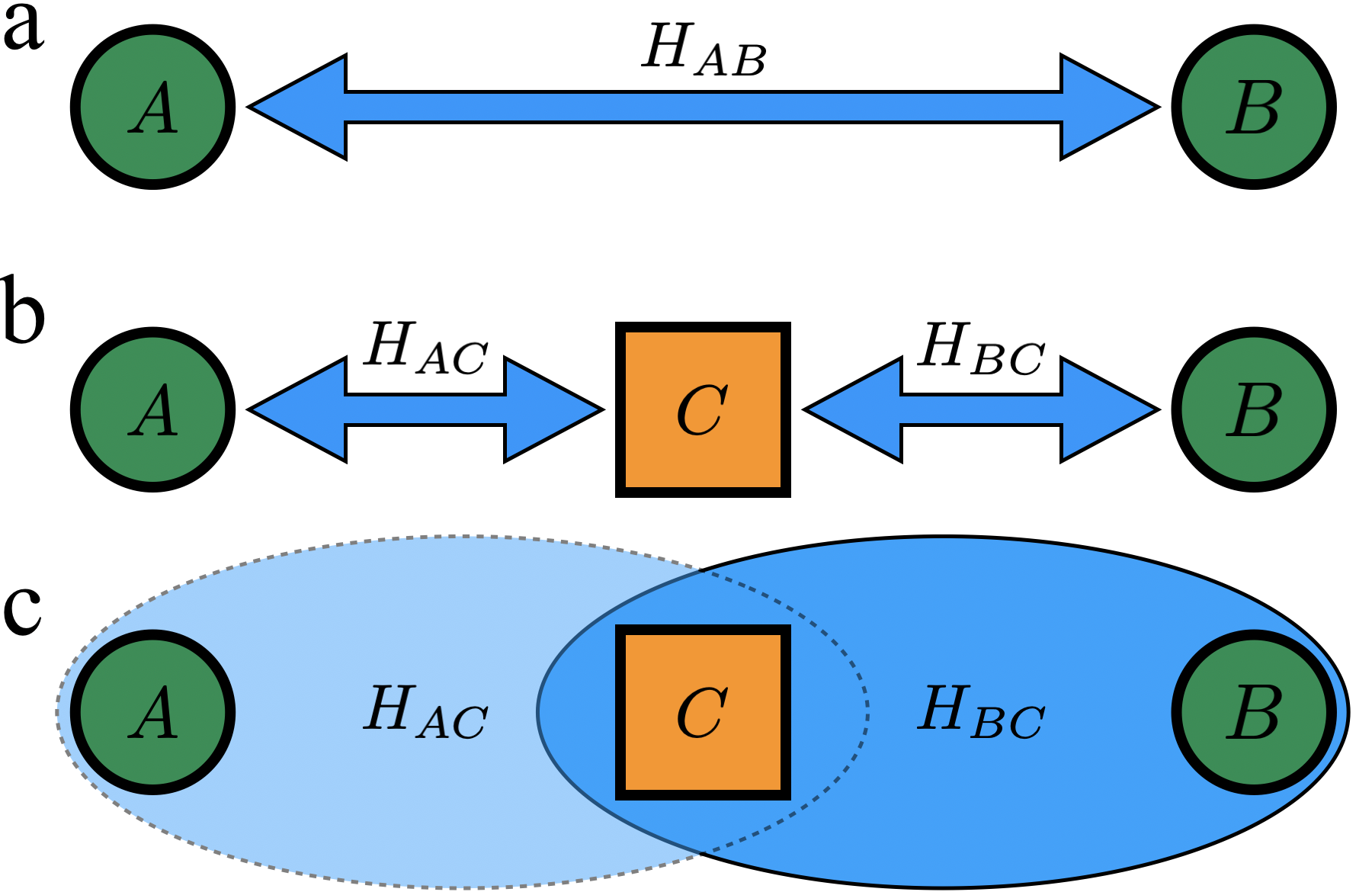}
\caption{Different considered scenarios.
The principal objects are denoted by $A$ and $B$.
Our goal is to maximally entangle them as fast as possible, starting with a disentangled initial state.
(a) Direct interactions, with Hamiltonian $H_{AB}$.
(b) Continuous mediated interactions with general Hamiltonians of the form $H_{AC}+H_{BC}$.
(c) Sequential mediated interactions where $C$ first interacts with $A$, and then with $B$.}
\label{FIG_setupspsp}
\end{figure}

Consider, for the moment, a unitary evolution of a quantum state $\rho(0)$ to $\rho_{\text{tar}}$ generated by a Hamiltonian $H$.
One can see from the unified QSL in Eq.~(\ref{EQ_bound}) that there are two relevant quantities: 
one being the fidelity $\mathcal{F}(\rho(0),\rho_{\text{tar}})$ between the initial and target state and
the other $\min \{ \langle H\rangle,\Delta H\}$, which is the minimum of the non-negative mean energy or SDE.
It is straightforward to check that scaling of the Hamiltonian, $H\rightarrow kH$, where $k$ is a constant, leads to the rescaled energy factors $\langle H\rangle \rightarrow k\langle H\rangle$ and $\Delta H\rightarrow k \Delta H$.
A trivial option to speed up the evolution of the quantum state is therefore to supply more energy, e.g., by having stronger coupling.
We wish to focus on other quantities playing a role in the speed of evolution and therefore, in what follows, we put the strength of all interactions on equal ground by setting $\min \{ \langle H\rangle,\Delta H\} = \hbar \Omega$, where $\Omega$ is a frequency unit. 
This allows us to write the unified QSL in Eq.~(\ref{EQ_bound}) as
\begin{equation}\label{EQ_dbound}
\Gamma(\rho(0),\rho_{\text{tar}})\ge \frac{\Theta(\rho(0),\rho_{\text{tar}})}{\min \{ \langle M\rangle,\Delta M\}},
\end{equation}
where $\Gamma=\Omega \tau$ stands for the dimensionless minimal time,
whereas $\langle M\rangle = \langle H \rangle / \hbar \Omega$ and $\Delta M = \Delta H / \hbar \Omega$ respectively denote the non-negative mean energy and SDE, normalised with respect to $\hbar \Omega$.
Hereafter, we assume the condition
\begin{equation}
\min \{\langle M\rangle,\Delta M\}=1,
\end{equation}
which can always be ensured with appropriate scaling $k$.
We refer to this condition as \emph{resource equality}.

To quantify the amount of entanglement, we use negativity, which is a well known computable entanglement monotone~\cite{neg98,neg99,negativity,neg2000,neg2005}.
Negativity is defined as the sum of negative eigenvalues after the state of a bipartite system is partially transposed.
The bipartite entanglement between objects $X$ and $Y$ is denoted by $N_{X:Y}$ and admits maximum value $(d-1)/2$, where $d=\min\{d_X,d_Y\}$ and $d_X$ ($d_Y$) is the dimension of object $X$ ($Y$).
For simplicity, we shall assume that the principal objects have the same dimension.
Maximally entangled states, for any entanglement monotone~\cite{streltsov2012general}, are given by pure states of the form
\begin{equation}
|\Psi_{XY} \rangle = \frac{1}{\sqrt{d}} \sum_{j = 0}^{d-1} |x_j \rangle |y_j \rangle,
\label{EQ_sp_msent}
\end{equation}
where $\{| x_j \rangle\}$ and $\{| y_j \rangle\}$ are orthonormal bases for object $X$ and $Y$, respectively.

%%%%%%%%%%%%%%%%%%%%%%%%%%%%%%%%%%%%%%%%%%%%%%%%%%%%%%%%%%%%%%%%%%%%%%%%%%%%%%%%%%%%%%%%%%%%%%%%%%%%%%%%%%%%%%%%%%

\section{Direct interactions}
Let us begin with optimal entangling dynamics for any dimension $d$, with direct interactions.
Since the initial state we take is disentangled, it has to be a pure product state as the dynamics is purity preserving and the final maximally entangled state is pure, see Eq.~(\ref{EQ_sp_msent}).
One easily verifies with the help of Cauchy-Schwarz inequality that the fidelity between a product state and maximally entangled state is bounded as 
$\mathcal{F} = \langle \alpha \beta |\Psi_{AB}\rangle \le 1 / \sqrt{d}$.
From the resource equality, the time to maximally entangle two systems via direct interactions follows
\begin{equation}\label{EQ_DITB}
\Gamma_{\mathcal{DI}}\ge \arccos(\mathcal{F})\ge \arccos\left(1/\sqrt{d}\right).
\end{equation}
This bound is tight and can be achieved with the following exemplary dynamics.
Under an initial state of $|00\rangle$, we take an optimal (to be shown below) Hamiltonian
\begin{equation}\label{EQ_xx}
H_{AB} = \frac{\hbar \Omega}{2\sqrt{d-1}}\: \sum_{j=1}^{d-1} (X_A^{j} + Y_A^{j}) \otimes (X_B^{j} + Y_B^{j}),
\end{equation}
where the subscripts indicate the corresponding system and we have defined $X^{j} \equiv |0\rangle \langle j|+|j\rangle \langle 0|$ and $Y^{j}\equiv -i|0\rangle \langle j|+i|j\rangle \langle 0|$.
Note that the constant factor ensures the resource equality.
%It is straightforward to generalise the above example to initial product state $|\alpha \beta \rangle$, with a change of basis for the operators in the Hamiltonian.
One can show that the state at time $t$ takes the form $\ket{\psi_{AB}(t)}=\cos(\Omega t)\ket{00}+\sin(\Omega t)(\sum_{j=1}^{d-1} \ket{jj})/\sqrt{d-1}$, 
and therefore it oscillates between the disentangled state $\ket{00}$ and a maximally entangled state $|\Psi_{AB} \rangle$.
The latter is achieved earliest at time $T\equiv \Omega t=\arccos(1/\sqrt{d})$, see Fig.~\ref{FIG_sp_exp1}.

\begin{figure}[h]
\centering
\includegraphics[width=0.45\textwidth]{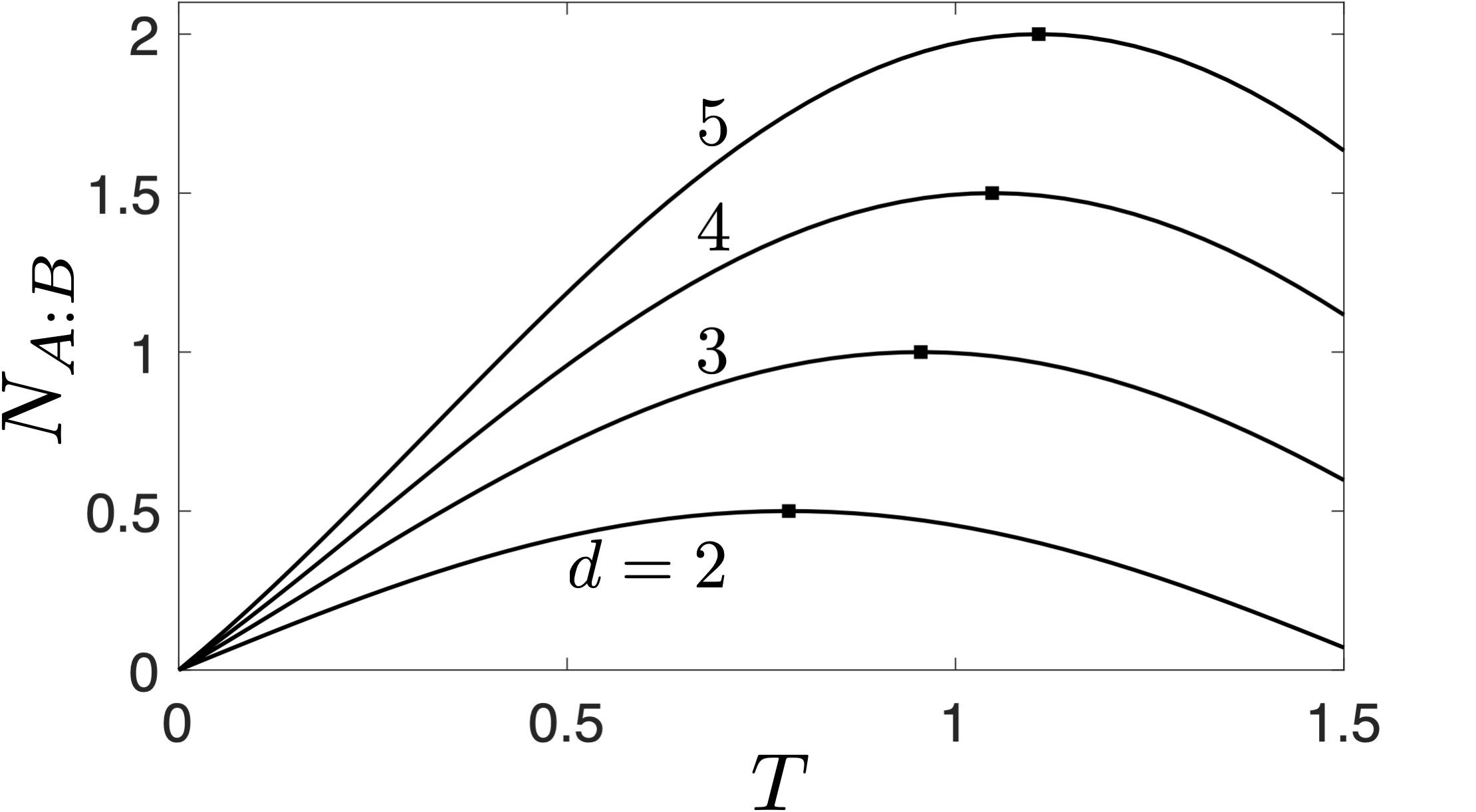}
\caption{Optimal direct dynamics showing maximum entangling speed between two objects, each with dimension $d$. 
Maximum entanglement, $(d-1)/2$, is achieved at $T=\arccos(1/\sqrt{d})$, indicated by the dots.
}
\label{FIG_sp_exp1}
\end{figure}

Alternatively, the optimality of this dynamics can be understood from the triangle inequality of the Bures angle~\cite{nielsenchuang}:
$\Theta(0,T)+\Theta(T,\arccos(1/\sqrt{d}))\ge \Theta(0,\arccos(1/\sqrt{d}))$, where we have used a short notation $\Theta(T_1,T_2)\equiv \Theta(\rho(T_1),\rho(T_2))$.
Under the resource equality, the optimal time should be equal to the Bures angle.
Indeed this is the case for the above dynamics as $\Theta(T_1,T_2) = T_2 - T_1$, \emph{saturating} the triangle inequality.
Therefore, not only the maximally entangled state is reached in the shortest time, but also all intermediate states as well.
%In other words, the negativity curves for arbitrary dynamics (up to maximally entangling time) all lie below the ones presented in Fig.~\ref{FIG_sp_exp1}. We shall return to this point later.

The described fastest entangling dynamics has the following special features.
(i) The Bures angle between any two states in the dynamics is proportional to entanglement gain,
so that QSL directly translates to the limits on entanglement generation.
(ii) This generation has its origin in components $(\sum_{j=1}^{d-1} \ket{jj})/\sqrt{d-1}$ and the high entangling speed comes from the fact 
that already the linear term in the expansion of the evolution operator $\exp{(-i \Delta t H_{AB}/\hbar)}$ introduces these components.
That is, the rate of change of entanglement is strictly positive $\dot N_{A:B}(t)>0$, for all times up to maximally entangling time.

%%%%%%%%%%%%%%%%%%%%%%%%%%%%%%%%%%%%%%%%%%%%%%%%%%%%%%%%%%%%%%%%%%%%%%%%%%%%%%%%%%%%%%%%%%%%%%%%%%%%%%%%%%%%%%%%%%

\section{Can mediator speed up entangling process?}
At first sight, one might wonder whether the use of quantum mechanical mediator could speed up the evolution by utilising non-commuting Hamiltonians, as revealed through the Baker-Campbell-Hausdorff (BCH) formula. Namely, the dynamics generated by direct coupling $H_{AB} = A \otimes B$ could be reconstructed through the mediator system $C$ interacting via $H_{AC} + H_{BC} = A \otimes p_C + x_C \otimes B$, where $x_C$ and $p_C$ are the position and momentum operators acting on the mediator.
Due to the canonical commutation relation the BCH equation reduces to:
\begin{eqnarray}
e^{-it(A \otimes p_C + x_C \otimes B)/ \hbar} & = &
e^{-it A \otimes p_C/ \hbar} \, 
e^{-it x_C \otimes B/ \hbar} \nonumber \\
& & e^{-it^2 A \otimes B/ 2 \hbar }.
\end{eqnarray}
Effective direct coupling is now identified in the last term on the right-hand side. 
Since the corresponding exponent is proportional to squared time, it is legitimate and interesting to enquire about the speeding up possibility.

On the other hand, the special features described at the end of the previous section make it unlikely that any other dynamics is faster than the fastest direct one.
Indeed, this is shown in Theorem~\ref{TH_untimate} presented in the Appendix.
Any dynamics (direct or mediated) that starts with disentangled principal systems
can maximally entangle them in time lower bounded as
\begin{equation}
\Gamma_{\text{any}} \ge \arccos \left( 1/\sqrt{d} \right),
\end{equation}
where the resource equality is assumed.
One then wonders whether mediated dynamics can achieve the same speed as the direct one.
At this stage initial correlations with the mediator enter the picture.

We shall show that if the mediator is initially completely uncorrelated from the principal systems, the time required to reach the maximally entangled state is \emph{strictly} larger than $\arccos(1/\sqrt{d})$.
Then we provide explicit examples of mediated dynamics, with initially correlated mediators, that achieve the shortest possible entangling time.

Consider the initial tripartite state of the form $\rho(0) = \rho_{AB} \otimes \rho_C$ (with separable $\rho_{AB}$) and, to give a vivid illustration first, take a Hamiltonian $H_{AC} + H_{BC} = (H_A + H_B) \otimes H_C$,
or any commuting Hamiltonians for which one can identify common eigenbasis $\{ |c \rangle \}$.
Let us take a specific product state $\ket{\alpha \beta \gamma}$ in the decomposition of the initial state $\rho(0)$, and write $\ket{\gamma} = \sum_c \lambda_c \ket{c}$.
Since $[ H_{AC}, H_{BC}] = 0$ the evolution is mathematically equivalent to $U_{BC} U_{AC}=\exp(-i t H_{BC}/\hbar)\exp(-i t H_{AC}/\hbar)$ and the initial product state evolves to 
$\ket{\psi(t)} = \sum_c \lambda_c |\alpha_c(t) \rangle |\beta_c(t) \rangle | c \rangle$,
where $|\alpha_c(t) \rangle = \exp(-i t E_c H_A/\hbar) \ket{\alpha}$ and $|\beta_c(t) \rangle = \exp(-i t E_c H_B/\hbar) \ket{\beta}$ with the corresponding eigenvalue $E_c$ of the Hamiltonian $H_C$.
By tracing out system $C$ we note that the state of $AB$ is a mixture of product states and hence not entangled.
Application of this argument to all the product states in the decomposition of $\rho(0)$ shows that this evolution cannot generate any entanglement between the principal systems whatsoever, i.e., $\Gamma_{\mathcal{CMI}} =\infty$ in this case.
This stark contrast with the QSL comes from the fact that the Bures angle is no longer related to the amount of entanglement in the subsystem $AB$.

Consider now a general Hamiltonian $H_{AC} + H_{BC}$.
In Theorem~\ref{TH_rate} presented in the Appendix we show that starting with $\rho(0) = \rho_{AB} \otimes \rho_C$ the mediated dynamics has non-positive entanglement rate at time $t=0$, 
i.e., $\dot N_{A:B}(0) \le 0$ if the three systems are open to their local environments and $\dot N_{A:B}(0) = 0$ for any closed mediated tripartite system.
This delay is causing a departure from the optimal entangling path and cannot be compensated in the future.
We show rigorously in Theorem~\ref{TH_CMI_strict} presented in the Appendix that starting with an uncorrelated mediator, i.e., $\rho(0) = \rho_{AB} \otimes \rho_C$ the time required to maximally entangle $A$ and $B$ via $\mathcal{CMI}$ satisfies a strict bound
\begin{equation}
    \Gamma_{\mathcal{CMI}}>\arccos{(1/\sqrt{d})}.
\end{equation}

Furthermore, we have performed numerical checks with random initial states and Hamiltonians (see the Appendix for details) and conjecture that the actual time to maximally entangle the principal systems with initially uncorrelated mediator is $\Gamma_{\mathrm{conj}} \ge 2 \arccos(1/\sqrt{d})$.
The following two examples with three quantum bits shed light on the origin of this hypothetical lower bound.
As initial state, consider $\ket{000}$, in the order $ABC$, and first take a Hamiltonian $H = \hbar \Omega (X_A Y_C + Y_B X_C)/\sqrt{2}$, where $X$ and $Y$ denote Pauli operators for the respective qubits.
One verifies that the resource equality holds and the state at time $t$ reads $\ket{\psi(t)} = \cos(\Omega t) \ket{000} + \sin(\Omega t) |\psi^+ \rangle \ket{1}$, where $|\psi^+ \rangle = (\ket{01} + \ket{10})/\sqrt{2}$ is the Bell state.
The maximally entangled state is obtained at time $\Omega t = \pi / 2$ because one has to wait until the dynamics completely erases the $\ket{000}$ component.
In contradistinction, the direct dynamics introduces $\ket{11}$ component (already in linear time $\Delta t$) and hence the evolution can stop at $\Omega t = \pi / 4$.
Another natural way to entangle two systems via mediator is to entangle the mediator with one of the systems first and then swap this entanglement.
Each of these processes takes time at least $\arccos(1/\sqrt{d})$ and hence again we arrive at the bound anticipated above (the swapping step actually takes a bit longer, see Appendix~\ref{APP_SMI}).
A rigorous proof of this bound is left as an open problem.

We finally give examples of mediated dynamics, starting with a correlated mediator, that entangles as fast as the fastest direct dynamics.
One may think of utilising an extreme option where the dynamics is initialised with a maximally entangled mediator.
This is indeed possible but it is also possible to utilise purely classical correlations with the mediator.
Let us start with the entangled mediator first.
Consider three qubits with an initial state and the Hamiltonian written as
\begin{eqnarray}\label{EQ_expone}
\ket{\psi(0)}&=&\frac{1}{\sqrt{2}}(\ket{000}+\ket{111}),\nonumber \\
H&=&\frac{\hbar \Omega}{2\sqrt{2}}(Z_A \otimes H_{C_1} + Z_B \otimes H_{C_2}),
\end{eqnarray}
where $H_{C_1}=-(\openone + X_C + Y_C + Z_C)$ and $H_{C_2} = \openone - X_C - Y_C + Z_C$.
The principal system is initially disentangled but the mediator is maximally entangled with the rest of the systems, $N_{AB:C}(0) = 1/2$.
One verifies that $N_{A:B}$ follows the curve for $d=2$ in Fig.~\ref{FIG_sp_exp1}.

As mentioned, quantum correlations to the mediator are not necessary.
Consider the following example:
\begin{eqnarray}\label{EQ_exp3}
\rho(0)&=&\frac{1}{2}\ket{\psi_m}\bra{\psi_m}\otimes \ket{0}\bra{0}+\frac{1}{2}|\tilde \psi_m\rangle \langle \tilde \psi_m|\otimes \ket{1}\bra{1},\nonumber \\
H&=&\frac{\hbar \Omega}{2}(Z_A \otimes Z_C + Z_B \otimes Z_C),
\end{eqnarray}
where $\ket{\psi_m}=(\ket{+-}+\ket{-+})/\sqrt{2}$ and $|\tilde \psi_m\rangle=(\ket{--}+\ket{++})/\sqrt{2}$ are two Bell-like states of $AB$ with $|\pm \rangle=(|0\rangle \pm |1\rangle)/\sqrt{2}$.
This example is similar to those in Refs.~\cite{krisnanda2017,pal2021experimental} used to demonstrate entanglement localisation via classical mediators and to indicate that controlled quantum teleportation can be realised without genuine multipartite entanglement~\cite{Barasinski2018}.
Note that initially the principal system is disentangled (an even mixture of Bell states) and this time the mediator is only classically correlated --- its states flag in which maximally entangled state is the principal system~\cite{flags}.
Furthermore, Hamiltonians $H_{AC}$ and $H_{BC}$ in Eq.~(\ref{EQ_exp3}) commute, with the common $Z$ eigenbasis, and hence in the absence of initial correlations with the mediator entanglement in the principal system would be impossible.
One can now verify via standard computations that the dynamics of $N_{A:B}$ resulting from Eq.~(\ref{EQ_exp3}) is the same as in Fig.~\ref{FIG_sp_exp1} for $d=2$.
Note that the states of the mediator are the eigenstates of $H$ and hence they are stationary.
Accordingly, only the Bell-like states evolve in time.
It has been shown recently in a general case of $\mathcal{CMI}$ where the state contains only classical correlations in the partition $AB:C$ at all times, 
that the entanglement gain, quantified by the relative entropy of entanglement~\cite{REEmutualinfo}, is bounded by the initial mutual information, i.e., $E_{A:B}(t)-E_{A:B}(0)\le I_{AB:C}(0)$~\cite{pal2021experimental}.
In the particular example of Eq.~(\ref{EQ_exp3}) this bound is achieved as we initially have $I_{AB:C}(0) = 1$ and $E_{A:B}(0) = 0$ which get converted to maximal entanglement $E_{A:B}(T) = 1$.
More generally, for the task discussed here an immediate strategy is to start with at least $I_{AB:C}(0)$ equal to the entanglement $E_{A:B}$ of the target state $|\Psi_{AB}\rangle$.

%%%%%%%%%%%%%%%%%%%%%%%%%%%%%%%%%%%%%%%%%%%%%%%%%%%%%%%%%%%%%%%%%%%%%%%%%%%%%%%%%%%%%%%%%%%%%%%%%%%%%%%%%%%%%%%%%%

\section{Sequential mediated interactions}
At last we move to the $\mathcal{SMI}$ scenario, where system $C$ first interacts only with $A$ and then only with $B$.
This setting was studied to some degree in Ref.~\cite{krisnanda2018detecting} where, in the present context, it was found that in order to prepare a maximally entangled state between the principal systems the dimension of $C$ has to be at least $d$.
We therefore set it to $d$ and take as initial state $\rho(0) = \rho_{AB} \otimes \rho_C$.
Under these conditions Theorem~\ref{TH_SMI} in the Appendix shows the following lower bound on the entangling time:
\begin{equation}
    \Gamma_{\mathcal{SMI}}\ge \arccos(1/\sqrt{d})+\arccos(1/d).
\end{equation}
Our numerical simulations indicate that this bound is tight.
Note that this is even longer than $2 \arccos(1/\sqrt{d})$ already demonstrated to be achievable with $\mathcal{CMI}$.

%%%%%%%%%%%%%%%%%%%%%%%%%%%%%%%%%%%%%%%%%%%%%%%%%%%%%%%%%%%%%%%%%%%%%%%%%%%%%%%%%%%%%%%%%%%%%%%%%%%%%%%%%%%%%%%%%%

\section{Discussion}
We wish to conclude with a few comments on the obtained results.
Since a maximally entangled state $|\Psi_{AB}\rangle$ is pure and the direct closed dynamics preserves the purity, the maximal entanglement cannot be achieved via direct coupling if one starts with a mixed state.
After introducing an ancillary system, the reduced $AB$ dynamics is, in general, not unitary and hence the purity of $\rho_{AB}$ may change.
For a concrete example see below Eq.~(\ref{EQ_exp3}), where the initial purity of $1/2$ is increased to $1$ while the disentangled initial state becomes maximally entangled.
Therefore, for states of $AB$ that are initially mixed, the only way to achieve maximum entanglement and saturate the time bound of $\mathcal{DI}$ is to make use of a correlated mediator.

Having said this, a possibility emerges to maximally entangle initially mixed principal systems by opening just one of them to a correlated local environment.
This is reasonable because the incoherent evolution may increase the purity of $\rho_{AB}$ and previously established entanglement with the environment can flow to the principal systems.
A simple example is as follows.
Suppose $A$ and $B$ are qubits and only qubit $A$ interacts with its single-qubit environment $C$.
As the initial state, we take the one in Eq.~(\ref{EQ_exp3}) and consider a Hamiltonian $H = \hbar \Omega \: Z_A \otimes Z_C$ for the local interaction with environment.
One verifies that the resulting dynamics gives the same entanglement $N_{A:B}$ as in Fig.~\ref{FIG_sp_exp1} for $d=2$.

The last example is interesting from the point of view of open quantum systems.
Note that the mutual information in the principal system grows from the initial value $I_{A:B}(0) = 1$ to the final value $I_{A:B}(\pi/4) = 2$.
Yet, subsystem $B$ has not been operated on --- only system $A$ interacts with its local environment.
One therefore asks what happens to the data processing inequality stating that information can only decay under local operations~\cite{nielsenchuang}.
The answer is that the inequality is derived for local maps which are completely positive and trace preserving.
Accordingly, the example just given is likely one of the simplest of non-completely-positive dynamics.
Violation of data processing inequality has already been discussed as a witness of such forms of evolution~\cite{Buscemi2014}.

Our main result shows that correlations play a similar role to energy in speeding up dynamics.
In tripartite mediated system A-C-B, where principal systems $A$ and $B$ are coupled via mediator $C$, it takes strictly longer to maximally entangle $AB$ when the evolution is initialised with uncorrelated mediator than when it begins with a correlated mediator.
We conjecture that the required minimal time for the case of uncorrelated mediator is twice as long.
In other words, if one would like to start with an uncorrelated mediator and reach a maximally entangled state at the same time as with a correlated mediator, one has to supply twice as much energy.

\begin{acknowledgements}
We thank Felix Binder and Varun Narasimhachar for stimulating discussions.
T.K. and T.C.H.L. acknowledge the support from the Ministry of Education (Singapore) project T2EP50121-0006.
C.N. was supported by the National Research Foundation of Korea (NRF) grant funded by the Korea government (MSIT) (NRF-2022R1F1A1063053).
J.K. was supported by KIAS Individual Grants (CG014604) at Korea Institute for Advanced Study.
A.S. was supported by the ``Quantum Coherence and Entanglement for Quantum Technology'' project, carried out within the First Team programme of the Foundation for Polish Science co-financed by the European Union under the European Regional Development Fund.
T.P. was supported by the Polish National Agency for Academic Exchange NAWA Project No. PPN/PPO/2018/1/00007/U/00001.
\end{acknowledgements}

\appendix

\section{No speeding up with mediators}

\begin{theorem}\label{TH_untimate}
Consider dynamics described by a Hamiltonian $H$, involving three objects $A$, $B$, and $C$ (direct or mediated).
For initial states $\rho(0)=\rho_{ABC}$, having disentangled $\rho_{AB}$, the lower bound on the time required to maximally entangle $AB$ satisfies
\begin{equation}
\Gamma_{\mathrm{any}}\ge \arccos \left( 1/\sqrt{d} \right),
\end{equation}
where the resource equality is assumed.
\end{theorem}
\begin{proof}
In the target state the principal systems are maximally entangled, which implies that their state is pure and uncorrelated with the mediator $C$, i.e., $\rho_{\text{tar}}=|\Psi_{AB}\rangle \langle \Psi_{AB}|\otimes \rho_C$.
We evaluate the fidelity of the initial and target states:
\begin{eqnarray}\label{EQ_utb}
\mathcal{F}(\rho(0),\rho_{\text{tar}})&=&\mathcal{F}(\rho_{ABC},\ket{\Psi_{AB}}\bra{\Psi_{AB}}\otimes \rho_C)\nonumber \\
&\le&\mathcal{F}(\rho_{AB},\ket{\Psi_{AB}}\bra{\Psi_{AB}}) \nonumber \\
%&=&\sqrt{\langle \Psi_{\text{max}} |\rho_{AB}|\Psi_{\text{max}}}\rangle \nonumber \\
&\le&\max_{p_j,\ket{a_jb_j}} \sqrt{\sum_j p_j |\langle a_jb_j|\Psi_{AB}\rangle|^2} \nonumber \\
&\le&\max_{\ket{a_jb_j}}  |\langle a_jb_j|\Psi_{AB}\rangle| =\frac{1}{\sqrt{d}},
\end{eqnarray}
where the steps are justified as follows.
The first inequality is due to monotonicity of fidelity under trace-preserving completely positive maps~\cite{nielsen1996entanglement} (here, tracing out $C$).
Then we expressed the disentangled state as $\rho_{AB}=\sum_j p_j\:\ket{a_jb_j}\bra{a_jb_j}$ and used its convexity properties.
The final equation follows from the form of maximally entangled state.
Finally, by having the resource equality, one gets $\Gamma_{\mathrm{any}}\ge \arccos\left({\mathcal{F}(\rho(0),\rho_{\text{tar}})}\right) \ge \arccos{(1/\sqrt{d})}$.
\end{proof}

\section{No initial entanglement gain with uncorrelated mediator}

\begin{theorem}\label{TH_rate}
Consider the case of $\mathcal{CMI}$, where all objects can be open to their own local environments (for generality).
For initial states where the mediator is uncorrelated, i.e., $\rho(0)=\rho_{AB}\otimes \rho_C$, the rate of any entanglement monotone follows $\dot E_{A:B}(0) \le 0$.
\end{theorem}

\begin{proof}
We take the evolution of the whole tripartite system following the Lindblad master equation to include the contribution from interactions with local environments:
\begin{eqnarray}
\frac{\rho(\Delta t)-\rho(0)}{\Delta t} & = & -i[H,\rho(0)]+\sum_{X=A,B,C}L_X\rho(0),\label{EQ_open} \\
L_X\rho(0) & \equiv & \sum_k Q^X_k\rho(0) Q^{X\dag}_k-\frac{1}{2}\{Q^{X\dag}_kQ^X_k,\rho(0)\}.  \nonumber 
\end{eqnarray}
We set $\hbar$ to unity in this proof for simplicity.
Note that the first term in the RHS of Eq.~(\ref{EQ_open}) corresponds to the coherent part of the dynamics, while the second constitutes incoherent processes from interactions with local environments, that is, the operator $Q^X_k$ only acts on system $X$.
We take the total Hamiltonian as $H=H_A\otimes H_C+H_B\otimes H_{C^{\prime}}$ without loss of generality, and note that the proof easily follows for a general Hamiltonian $H=\sum_{\mu} H_A^{\mu}\otimes H_C^{\mu}+\sum_{\nu} H_B^{\nu}\otimes H^{\nu}_{C^{\prime}}$.

Following Eq.~(\ref{EQ_open}), the state of the principal objects at $\Delta t$ reads
\begin{eqnarray}
\rho_{AB}(\Delta t)&=&\mbox{tr}_C(\rho(\Delta t))\nonumber \\
&=&\mbox{tr}_C( \rho(0)-i\Delta t [H,\rho(0)]+\Delta t \sum_{X}L_X\rho(0) )\nonumber \\
&=&\rho_{AB}-i\Delta t[H_AE_C+H_BE_{C^{\prime}},\rho_{AB}]\nonumber \\
&&+\Delta t(L_A+L_B)\rho_{AB}, \label{EQ_rabdt}
\end{eqnarray}
where $E_C=\mbox{tr}(H_C\rho_C)$ and $E_{C^{\prime}}=\mbox{tr}(H_{C^{\prime}}\rho_C)$ denote the initial mean energies, and we have used $\rho(0)=\rho_{AB}\otimes \rho_C$. 
Also, $\mbox{tr}_C(Q^C_k\rho_C Q^{C\dag}_k-\frac{1}{2}\{Q^{C\dag}_kQ^C_k,\rho_C\})=0$ follows from the cyclic property of trace. 

Effectively, the evolution of the principal objects leading to $\rho_{AB}(\Delta t)$, as written in Eq.~(\ref{EQ_rabdt}), consists of local Hamiltonians weighted by the corresponding mean energies $H_AE_C+H_B E_{C^{\prime}}$,
and interactions with respective local environments.
Therefore, for any entanglement monotone, a measure that is non-increasing under local operations and classical communication, one concludes that $E_{A:B}(\Delta t)\le E_{A:B}(0)$, and hence, $\dot E_{A:B}(0)\le 0$.
In particular, this holds for negativity used in the main text.
\end{proof}

Unitary dynamics is a special case of Theorem~\ref{TH_rate} without incoherent interactions with local environments.
Since entanglement monotones are invariant under local unitary operations $E_{A:B}(\Delta t)= E_{A:B}(0)$ or $\dot E_{A:B}(0)= 0$.
As a consequence, changes in entanglement between the principal objects (positive or negative) are only possible if the mediator $C$ is correlated with them.

By applying this argument to the final state $| \Psi_{AB} \rangle \langle \Psi_{AB}| \otimes \rho_C$ and backwards in time,
we conclude that any dynamics (direct or mediated) approaches the final state at a rate $\dot E_{A:B}(T)= 0$, clearly seen in Fig.~\ref{FIG_sp_exp1}.

\section{Strict bound for uncorrelated mediator}

We revisit the condition where $C$ is initially uncorrelated, i.e., $\rho(0)=\rho_{AB}\otimes \rho_C$, which is a special case of Theorem~\ref{TH_untimate}.
In this case, we have
\begin{equation}\label{EQ_tbspecial}
\mathcal{F}(\rho(0),\rho_{\text{tar}})=\mathcal{F}(\rho_{AB},\ket{\Psi_{AB}}\bra{\Psi_{AB}})\:\mathcal{F}(\rho_{C},\rho_C^{\prime}),
%&&\text{tr} \left( \sqrt{\sqrt{\rho_{AB}\rho_C} \ket{\Psi_{\text{max}}}\bra{\Psi_{\text{max}}} \rho_C^{\prime} \sqrt{\rho_{AB}\rho_C}  }\right) \nonumber \\
%=\text{tr} \left( \sqrt{\sqrt{\rho_{AB}} \ket{\Psi_{\text{max}}}\bra{\Psi_{\text{max}}}\sqrt{\rho_{AB}}  \sqrt{\rho_C} \rho_C^{\prime}\sqrt{\rho_C}}\right) &&\nonumber \\
%&=&\text{tr} \left( \sqrt{\sqrt{\rho_{AB}} \ket{\Psi_{\text{max}}}\bra{\Psi_{\text{max}}}\sqrt{\rho_{AB}} }\right) \text{tr} \left( \sqrt{\sqrt{\rho_C} \rho_C^{\prime}\sqrt{\rho_C}}\right)\nonumber \\
%&&=\mathcal{F}(\rho_{AB},\ket{\Psi_{\text{max}}}\bra{\Psi_{\text{max}}})\:\mathcal{F}(\rho_{C},\rho_C^{\prime}).
\end{equation}
%\end{widetext}
where $\rho_C^{\prime}$ is the state of $C$ in the target $\rho_{\text{tar}}$.
The only way to saturate the optimal bound of direct dynamics is to set $\mathcal{F}(\rho_{C},\rho_C^{\prime}) = 1$, i.e. $\rho_C = \rho_C^{\prime}$.
Accordingly, the initial state of $AB$ has to be in a pure product form.
%Note that the first inequality of Eq.~(\ref{EQ_utb}) is recovered as $\mathcal{F}(\rho_{C},\rho_C^{\prime})\in [0,1]$.
%Now, suppose that the initial state $\rho_{AB}$ is strictly mixed.
%As the dynamics is unitary (preserving purity), to obtain a pure target state on the principal systems, the final state of $C$ has to be strictly more mixed than its initial one, and therefore $\rho_C\ne \rho_C^{\prime}$.
%Consequently, we have $\mathcal{F}(\rho_C,\rho_C^{\prime})<1$.
%One can then follow the proof of Theorem~\ref{TH_untimate} and have a strict sign for the first inequality of Eq.~(\ref{EQ_utb}), resulting in a strict bound $\Gamma_{\mathcal{CMI}}>\arccos(1/\sqrt{d})$. 
%This implies that the possibility of saturating the optimum dynamics requires the initial state $\rho_{AB}$ to be pure and $\mathcal{F}(\rho_{C},\rho_C^{\prime})=1$, i.e., $\rho_C=\rho_C^{\prime}$.
%A trivial dynamics (via direct interactions in a tripartite setting) is given by the example presented in Section~\ref{SC_dents} with the addition of a decoupled mediator $\rho_C$ and a local Hamiltonian $H_C$.
Having this in mind, the theorem below shows that the time bound is still strict.

\begin{theorem}\label{TH_CMI_strict}
For the initial state of the form
%of $AB$ is pure without entanglement and that it is not correlated with the mediating system $C$, i.e.,
$\rho(0)=|\alpha \beta \rangle \langle \alpha \beta |\otimes \rho_C$, the time required to maximally entangle the principal systems via $\mathcal{CMI}$ follows a strict bound
\begin{equation}
    \Gamma_{\mathcal{CMI}} > \arccos(1/\sqrt{d}).
\end{equation}
\end{theorem}
\begin{proof}
    Recall that the dynamics identified in the $\mathcal{DI}$ case saturates the triangle inequality and is characterised by a straight line in Bures angles.
    Any other optimal dynamics (e.g. generated by other Hamiltonians) has to follow the same straight line. Along the line the states of $AB$ remain pure at all times.
    However, Theorem~\ref{TH_rate} shows that entanglement gain between $A$ and $B$ is possible only when 
    the mediating system is correlated with the principal systems at some time $t$ during the dynamics. 
    In the present case, this means that at $t$, the state of $AB$ is not pure, in particular, the mediator is not in a decoupled form $|\psi_{AB}(t)\rangle \langle \psi_{AB}(t)|\otimes \rho_C$, where $|\psi_{AB}(t)\rangle$ is the state from the optimum $\mathcal{DI}$. 
    Since $\mathcal{F}(\rho_{AB}(t),|\psi_{AB}(t) \rangle \langle \psi_{AB}(t)|) < 1$, we use the triangle inequality of the Bures angle to conclude the strict bound: 
    \begin{eqnarray}
        \Gamma_{\mathcal{CMI}}&=&\Gamma_1+\Gamma_2\nonumber \\
        &\ge& \Theta(0,t)+\Theta(t,\arccos(1/\sqrt{d}))\nonumber \\
        &>& \Theta(0,\arccos(1/\sqrt{d}))=\arccos(1/\sqrt{d}),
    \end{eqnarray}
    where $\Gamma_1$ and $\Gamma_2$ respectively denote the minimum time for evolution $0\rightarrow t$ and $t\rightarrow \arccos(1/\sqrt{d})$.
    In other words, the dynamics strictly does not follow the optimum (straight line) path, where at time $t$ the state is uniquely $|\psi_{AB}(t)\rangle \langle \psi_{AB}(t)|\otimes \rho_C$.
\end{proof}

\section{Numerical simulations for uncorrelated mediator}
Here we present results of numerical simulations behind the conjectured minimal time of $2\arccos(1/\sqrt{d})$ to maximally entangle the principal systems with initially uncorrelated mediator.
Based on the discussion prior to Theorem~\ref{TH_CMI_strict}, we consider initial states of the form $\rho(0)=|\alpha \beta \rangle \langle \alpha \beta|\otimes \rho_C$ and
Hamiltonians $H=H_{AC}+H_{BC}$ scaled to satisfy the resource equality condition.

Let us first describe the case of three qubits. 
We random the initial state, i.e., $|\alpha \rangle$, $|\beta \rangle$, and $\rho_C$ as well as the Hamiltonians $H_{AC}$ and $H_{BC}$ using the quantinf package by Toby Cubitt. 
Fig.~\ref{FIG_smi_num} presents entanglement generated by such evolutions, with $10^6$ random instances at each time.
As seen, the fastest time to reach maximum entanglement of $0.5$ is indeed $2\arccos(1/\sqrt{2})$.
We also performed simulations for the same number of random instances for three qutrits ($d=3$).
In this case, entanglement does not even come close to the maximum possible value at time $2\arccos(1/\sqrt{3})$, indicating that this is a correct lower bound on the entangling time.
\begin{figure}[h]
\centering
\includegraphics[width=0.45\textwidth]{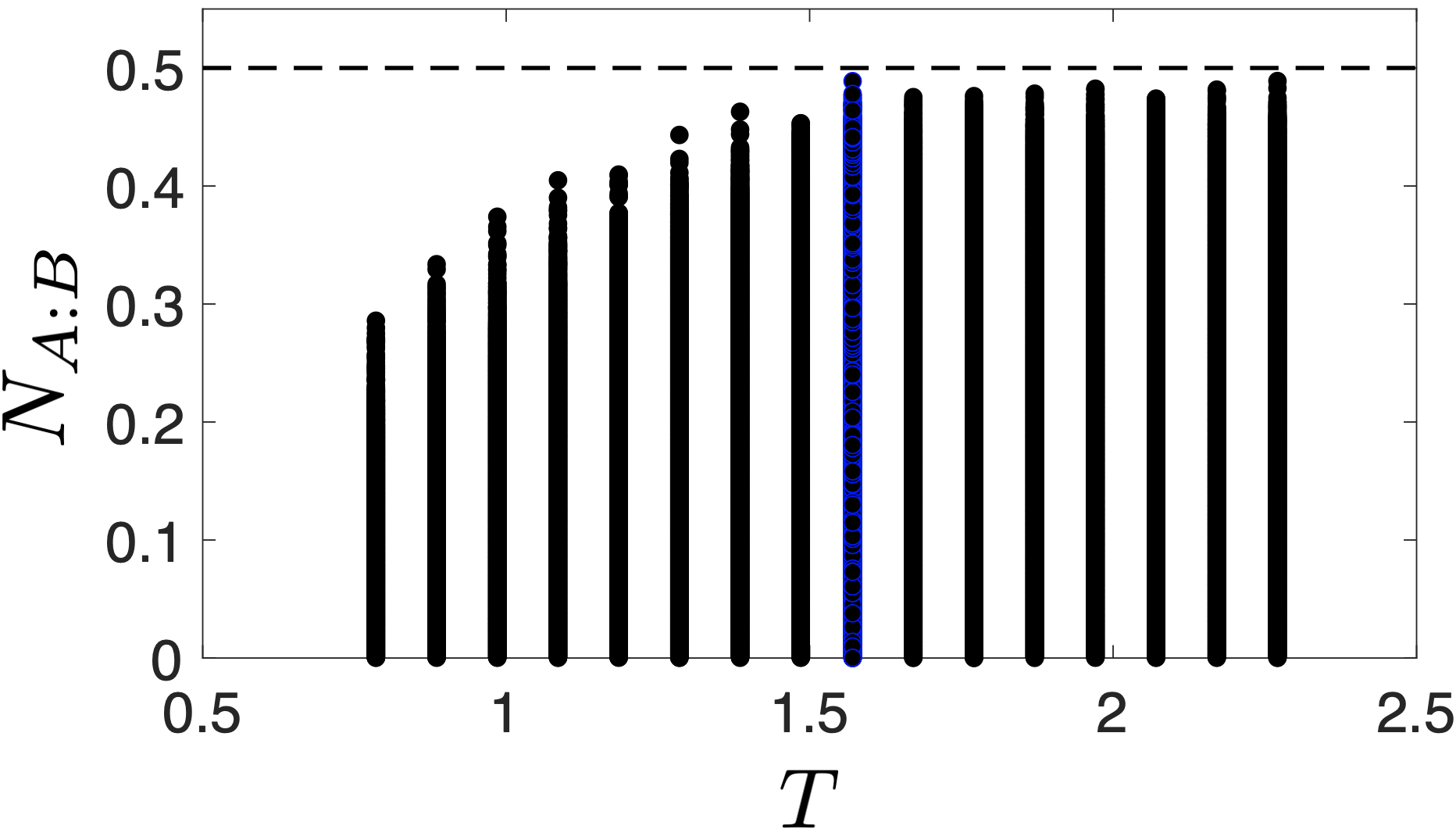}
\caption{
Numerical simulations of $\mathcal{CMI}$ with three qubits and initially uncorrelated mediator.
We generated $10^6$ random initial states and Hamiltonians for each time.
The points highlighted in blue correspond to evolution time $2\arccos(1/\sqrt{2})$.
The dashed line indicates maximum entanglement between two qubits.
}
\label{FIG_smi_num}
\end{figure}

\section{Sequential mediated dynamics}
\label{APP_SMI}

\begin{theorem}
\label{TH_SMI}
Starting with $\rho(0) = \rho_{AB} \otimes \rho_C$, maximal entanglement in AB is achieved via $\mathcal{SMI}$ in time
\begin{equation}
    \Gamma_{\mathcal{SMI}}\ge \arccos(1/\sqrt{d})+\arccos(1/d).
\end{equation}
\end{theorem}

\begin{proof}
The final state has the form $\rho_f = | \Psi_{AB} \rangle \langle \Psi_{AB} | \otimes \rho_C$.
In this scenario it is to be obtained by the sequence of operations $\rho_f = U_{BC} U_{AC} \rho(0) U_{AC}^\dagger U_{BC}^\dagger$.
We start with the following argument
\begin{equation}
E_{A:B}(\rho_f) \le E_{A:BC}(\rho_f) = E_{A:BC} ( \, U_{AC} \rho(0) U_{AC}^\dagger \, )
\end{equation}
where the inequality is due to the monotonicity of entanglement under local operations (here, tracing out $C$) 
and the equality is due to the fact that the second unitary, $U_{BC}$, is local in the considered bipartition.
Thus the only way to establish maximal final entanglement between the principal systems is to already prepare it with operation $U_{AC}$.
This consumes time $\arccos(1/\sqrt{d})$ and requires initial state of $A$ and $C$ to be pure, i.e. $\ket{\alpha \gamma}$ because $C$ is not correlated with $AB$ initially (note that it does not pay off to start with partial entanglement in $\rho_{AB}$).
Furthermore, since the final state is pure and we are left with application of $U_{BC}$ only, the state of particle $B$ also has to be pure.
Summing up, after the first step the tripartite state reads $| \Psi_{AC} \rangle \ket{\beta}$.
In the remaining step we need to swap this maximal entanglement into the principal systems.
To estimate the time required by the swapping we compute the fidelity:
\begin{eqnarray}
    \mathcal{F}&=&|\langle \Psi_{AC}|\langle \beta|\Psi_{AB}\rangle |\gamma \rangle |\nonumber \\
    &=&\frac{1}{d}|\sum_{j=1}^d\sum_{k=1}^d \langle a_j|a^{\prime}_k \rangle \langle \beta|b^{\prime}_k \rangle \langle c_j|\gamma \rangle| \nonumber \\
    &\le&\frac{1}{d}\sqrt{\sum_j|\sum_k \langle a_j|a^{\prime}_k \rangle \langle \beta|b^{\prime}_k \rangle|^2}\sqrt{\sum_j |\langle \gamma|c_j \rangle|^2}\nonumber \\
%    &= &\frac{1}{d}\sqrt{\sum_j (\sum_k \langle a_j|a^{\prime}_k\rangle\langle \beta|b^{\prime}_k \rangle)(\sum_l \langle a^{\prime}_l|a_j \rangle\langle b^{\prime}_l|\beta  \rangle)}\nonumber \\
    &= &\frac{1}{d}\sqrt{\sum_{j,k,l} \langle a^{\prime}_l|a_j \rangle \langle a_j|a^{\prime}_k\rangle \langle \beta|b^{\prime}_k \rangle \langle b^{\prime}_l|\beta  \rangle} = \frac{1}{d},
\end{eqnarray}
where we have written $|\Psi_{AC}\rangle =\sum_j |a_jc_j\rangle/\sqrt{d}$ and $|\Psi_{AB}\rangle =\sum_k |a^{\prime}_kb^{\prime}_k\rangle/\sqrt{d}$ as the maximally entangled states (note possibly different Schmidt bases). 
Then we used the Cauchy-Schwarz inequality to obtain the third line.
Since $\{|c_j\rangle \}$ form a complete basis the last square root in the third line equals $1$ (sum of probabilities).
Rewriting the remaining mod-squared and using the completeness of the bases $\{|a_j\rangle \}$ and $\{|b_k^{\prime}\rangle \}$ we arrive at the final result. 
The total time required by both steps is therefore at least $\Gamma_{\mathcal{SMI}}=\Gamma_1+\Gamma_2\ge \arccos(1/\sqrt{d})+\arccos(1/d)$.
\end{proof}

%\lhead{\emph{Bibliography}} % Change the page header to say "Bibliography"
%\addcontentsline{toc}{chapter}{Bibliography}
%\bibliographystyle{unsrtnat} % Use the "unsrtnat" BibTeX style for formatting the Bibliography
\bibliography{Bibliography}

%\clearpage

\end{document}